\documentclass[conference]{IEEEtran}     

\IEEEoverridecommandlockouts                              

\usepackage{graphics} 
\usepackage{epsfig} 
\usepackage{times} 
\usepackage{amsmath,bbm,amsfonts,amssymb,amsthm}  
\usepackage{url}
\usepackage{dsfont}
\newtheorem{theorem}{Theorem}

\usepackage{cite,comment}

\usepackage{color}
\usepackage{algorithm}

{ 

}
\usepackage{multirow}

\newcommand{\bi}{\begin{itemize}}
\newcommand{\ei}{\end{itemize}}
\newcommand{\bd}{\begin{displaymath}}
\newcommand{\ed}{\end{displaymath}}
\newcommand{\be}{\begin{eqnarray*}}
\newcommand{\ee}{\end{eqnarray*}}

\newcommand{\Pfdroop}{\text{$P$-\,$\omega$ }}
\newcommand{\QVdroop}{\text{$Q$-$V$ }}

\usepackage[normalem]{ulem}

\title{
Coordinated Frequency Regulation in Grid-Forming Storage Network via Safety-Consensus }
\author{ Ramij Raja Hossain, Kaustav	Chatterjee, Sai Pushpak	Nandanoori, Soumya	Kundu, \\ Laurentiu Marinovici, Karan Kalsi, and Diane Baldwin
\thanks{The authors are with Pacific Northwest National Laboratory (PNNL), Richland, 99354, WA, USA. This research was supported by the Embedded Storage Project funded by the DOE Office of Electricity (OE). PNNL is a multi-program national laboratory operated for the U.S. Department of Energy (DOE) by Battelle Memorial Institute under Contract No. DE-AC05-76RL01830.}}

\begin{document}
\maketitle
\begin{abstract}
Inverter-based storages are poised to play a prominent role in future power grids with massive renewable generation. Grid-forming inverters (GFMs) are emerging as a dominant technology with synchronous generators (SG)-like characteristics through primary control loops.
Advanced secondary control schemes, e.g., consensus algorithms, allow GFM-interfaced storage units to participate in frequency regulations and restore nominal frequency following grid disturbances. However, it is imperative to ensure transient frequency excursions do not violate critical safety limits while the grid transitions from pre- to post-disturbance operating point. This paper presents a hierarchical safety-enforced consensus method --- combining a device-layer (decentralized) transient safety filter with a secondary-layer (distributed) consensus coordination --- to achieve three distinct objectives: limiting transient frequency excursions to safe limits, minimizing frequency deviations from nominal, and ensuring coordinated power sharing among GFM-storage units. The proposed hierarchical (two-layered) safety-consensus technique is illustrated using a GFM-interfaced storage network on an IEEE 68-bus system under multiple grid transient scenarios. 

\end{abstract}



\section{Introduction}
To meet the pressing challenge of the low carbon future, power utilities around the world have started integrating large-scale renewable generation. However, the associated volatility and intermittency of renewable power resources necessitate the inclusion of energy storage units as a key grid infrastructure \cite{o2019use,twitchell2022enabling}. This has led to the emerging concept of an \textit{embedded storage network} with grid-connected storage units --- strategically placed at the transmission-distribution intersections --- providing core grid support, serving as energy buffers, and enhancing system flexibility and resilience under varying operating conditions \cite{o2019use,kerby2024impacts,chatterjee2024grid}. In order to achieve the full potential of such network-embedded bi-directional storage units in providing essential grid support under disturbances, it is imperative to design advanced control strategies for system-wide coordination of the storage units.

Recent advances in inverter technologies have led to the development of grid-forming (GFM) technology --- which acts as a controlled voltage source, participates in frequency/voltage regulations, and mimics certain characteristics of synchronous generators (SGs) --- thereby emerging as a widely sought-after solution for grid-connecting storage and renewable energy sources \cite{lasseter2019grid}.
%
%
%
%
%
%
%
%
  %
%
The control layers of GFMs are usually characterized by hierarchical structure \cite{bidram2012hierarchical,9290331}, with primary controls at the individual device level and secondary controls at the system level.
Local feedback-based $P$-$f$ and $Q$-$V$ droop laws constitute a popular primary control mechanism in GFM, providing frequency/voltage regulation and power-sharing at a faster time-scale with grid dynamics \cite{de2007voltage,du_wei_mod}. 
Secondary controls, on the other hand, are needed to eliminate the steady-state deviations of frequency/voltage from its nominal value and usually operate at a relatively slower time-scale \cite{bidram2012hierarchical,porco}. Of particular interest to this work are the distributed (especially, consensus-based) approaches to secondary control, \cite{porco,mohiuddin2020unified,singhal2022consensus,9732645}, which are attractive because of their improved scalability, and resilience against single-point attack/failure. The work in \cite{porco}, for example, combines decentralized proportional droop control and integral control with distributed averaging in secondary voltage and frequency control of islanded microgrid, while \cite{mohiuddin2020unified,singhal2022consensus,9732645} apply consensus control to coordinate GFMs and GFLs in microgrids.

It is imperative that critical safety limits (such as frequency limits as per IEEE Standards \cite{ieee1547}) are not violated during grid operations, at any time. The state-of-the-art primary and secondary control strategies of GFMs, however, do not explicitly account for safety guarantees in their design. Therefore, as we illustrate later in this paper, these existing controls cannot prevent unsafe transient excursions as the grid transitions to a post-disturbance operating point. The idea of safety-enforced control design in power grid, as explored recently in \cite{kundu2019distributed,chen2019compositional,kundu2020transient,bouvier2022distributed,sun2023learning,slac3r}, deals with steering the system transient trajectories to stay within (or, converge to) pre-specified constraints or bounds. In these works, barrier functions-based methods are used to certify the forward invariance of a pre-specified safe operating region, thereby guaranteeing safe transient excursions. However, the safety control methods in  \cite{kundu2019distributed,chen2019compositional,kundu2020transient,bouvier2022distributed,sun2023learning,slac3r} have been proposed as single-layer, local (decentralized or distributed), solutions and often rely on computationally expensive optimizations which are not amenable to real-time implementations.


In this work, we address this gap in the existing state-of-the-art control mechanisms by proposing a safety-enforced, hierarchical, distributed control strategy for GFM-storage network. Specifically, a local measurement-driven, minimally invasive, control solution --- henceforth referred to as `safety-consensus' --- is proposed to augment the local safety controls with the existing consensus-based secondary-layer. The proposed `safety-consensus' method achieve three distinct objectives: (a) safe transient frequency evolution, (b) minimizing frequency deviation, and (c) coordinated power sharing. In this paper, we consider a network of GFM-interfaced storage units, equipped with droop-controls at the primary-layer and a consensus-control at the secondary-layer. We implement and validate the proposed   `safety-consensus' control solution with the GFM-storage network on a modified IEEE 68-bus system for optimal, coordinated, and safe system operation under multiple grid transient scenarios. 

\section{System Modeling}
We consider a positive sequence model of a transmission network with $N$ buses, $m$ SGs, and $n$ GFM-storage units. The active and reactive power injections at bus $i$ are:
\begin{subequations}\label{pfeqn}
\begin{align}
    & P_i \!=\!V_i\sum_{k=1}^{N}V_k\Big[G_{ik}\cos{(\theta_i-\theta_k)} \!+\! B_{ik}\sin{(\theta_i-\theta_k)}\Big],\\
   & Q_i \!=\!V_i\sum_{k=1}^{N}V_k\Big[G_{ik}\sin{(\theta_i-\theta_k)} \!- \!B_{ik}\cos{(\theta_i-\theta_k)}\Big]
\end{align}
\end{subequations}
where $G_{ik}$ and $B_{ik}$ are, respectively, the transfer conductance and susceptance of the line connecting the buses $i$ and $k$.
The SGs are modeled with the generator-governor system \cite{pai}:
\begin{subequations}
\label{eq:sg}
\begin{align}
\text{(SG-$i$)}\qquad &\dot{\theta}_i = \omega_i - \omega_0,\\
&\dot{\omega}_i = \frac{1}{M_i}\Big[D_i(\omega_0-\omega_i) + P_{mi} - P_{i}\Big],\\
&\dot{P}_{mi} = -\frac{1}{T_{ch}}\Big[P_{mi} + \frac{\omega - \omega_0}{R_{gov}}\Big].
\end{align} 
\end{subequations}
\noindent where $\theta_i, \omega_i$, $\omega_0$, $M_i$, $D_i$, $P_{mi}$, and $P_i$ are generator angle, frequency (rad/s), nominal frequency (rad/s), inertia, damping coefficients, mechanical power input, and power injection, respectively. 
$T_{ch}$ and $R_{gov}$ are, respectively, the time-constants and droop of the speed governor. 
Following the concept of the embedded storage network \cite{kerby2024impacts,chatterjee2024grid}, GFM-storage units are located close to each of the $n$ load buses. GFMs adopt the WECC-approved REGFM\_A1 model \cite{du2023model} with \Pfdroop and \QVdroop droop, described as follows \cite{kwon2023risk,kwon2024coherency}:
\begin{subequations}
\label{eq:gfm}
\begin{align}
\!\!\text{(GFM-$i$)}\quad &\dot{\theta}_i = \omega_i - \omega_0,\\
&\dot{\omega}_i = \frac{1}{\tau_i}\big[\omega_0 - \omega_i +m_{p_i} (P^{{set}}_i - P_i)\big],\label{omega_eq}\\
&\dot{V}^e_i \!= \!\frac{1}{\tau_i}\big[V_i^{set} \!-\! V_i \!-\! {V}^e_i \!+\! m_{q_i} (Q^{{set}}_i \!-\! Q_{i})\big]
,\\ &\dot{E}_i = k^{pv}_i \dot{V}^e_i + k^{iv}_i V^e_i\,,\\
& S_i\leq \sqrt{P_i^2+Q_i^2}
\end{align} 
\end{subequations}
where, $\theta_i$ and $\omega_i$ are, respectively, the voltage angle and frequency of the internal bus, $V_i$, and $E_i$ are respectively the voltage magnitudes of the external, and the internal buses. $m_{p_i}, m_{q_i}$ are the \Pfdroop and \QVdroop droop gains, respectively. $\tau_i$ is the time constants of the low-pass measurement filter. $k^{pv}_j$ and $k^{iv}_j$ are, respectively, the proportional and integral gains in the \QVdroop droop control. $V_i^e$ represents voltage error. $P_i^{set}$, $Q_i^{set}$ and $V_i^{set}$ are the active power, reactive power, and voltage set-points. The injected active and reactive power ($P_i,Q_i$) from the GFM-storage satisfy the capacity constraint, where $S_i$ is the apparent power rating (capacity) of the GFM-storage.
We consider the storage units to remain dormant during equilibrium conditions and come into action after system disturbances. Therefore, in pre-disturbance steady-state condition, for $i$th storage device: $P_i^{set} = 0$, and $Q_i^{set} = 0$. 

\section{Proposed method: Safety-Consensus Control}

\begin{figure*}[t]
    \centering
    \includegraphics[width = 0.80\linewidth]{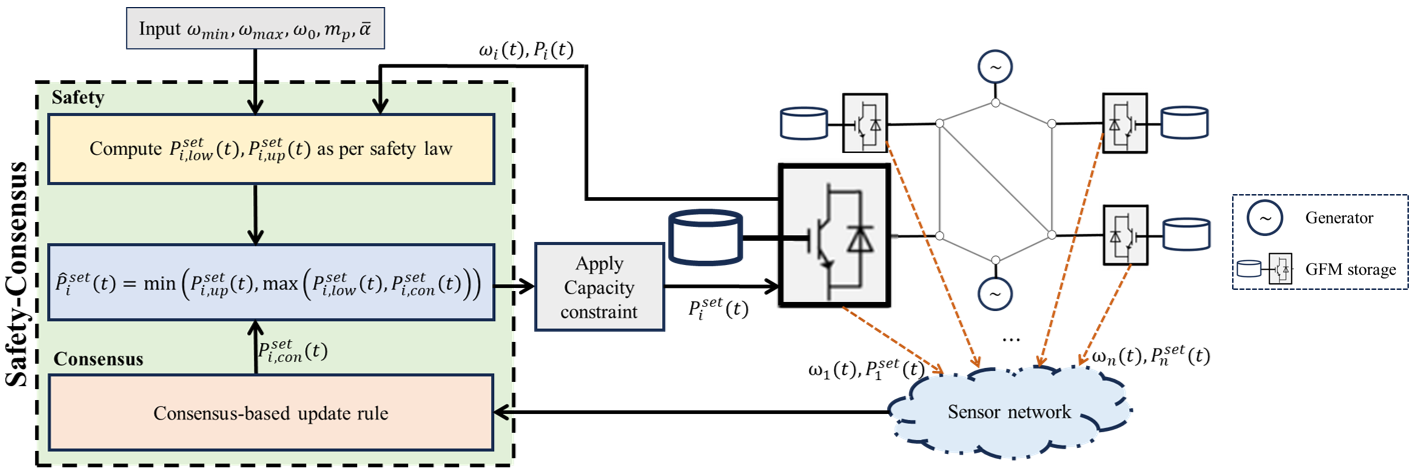} 
    \caption{Proposed \text{`safety-consensus'} method augments a local, fast timescale, safety-control with a distributed, slower timescale, consensus-control. }
    \label{fig:framework}
    \vspace{-8pt}
\end{figure*}

\subsection{Conventional Hierarchical Control of GFMs}
GFM-interfaced storage units are equipped with two essential control layers: (a) a primary control layer and (b) a secondary control layer, to achieve satisfactory steady-state and transient system operation under varying conditions. 
For the GFM model (\ref{eq:gfm}) used in this work, the primary control of GFM-interfaced storage units comes from the droop functions and depends on the droop coefficients $m_p$ and $m_q$. This primary control layer operates in faster time-scale (millisecond (ms) timescales) and helps to participate in frequency/voltage regulations by dispatching active/reactive power proportional to droop (in p.u.) times frequency/voltage deviations. Here, it is important to note that set-point values $P_i^{set}$, $Q_i^{set}$ remain unaffected in exercising primary control, resulting in post-disturbance steady-state frequency/voltage deviations. 

The secondary control layer, on the other hand, operates on a slower timescale (in the order of a few seconds) and updates the set-points $P_i^{set}$, $Q_i^{set}$ to correct the steady-state deviations of system variables. In this paper, we focus on GFM-interfaced storage units participating in frequency regulation (while a similar approach could be extended for voltage regulation, as well). As per the consensus-based algorithm in \cite{singhal2022consensus}, the update rule for $P_i^{set}$ can be written as follows:
\begin{align}
        \!P_{i,+}^{set}  &= P_{i}^{set}\!- \zeta_1 \!({\omega_{i}} \!-\!\omega_0) \!- \!\zeta_2\! \sum_{j \in \mathcal{N}_i} (m_{p_i} P^{set}_{i} \!\!-\! m_{p_j} P^{set}_{j}) \label{consPset}
\end{align}
where the subscript `$+$' is used to denote the update to the set-point $P_i^{set}$ based on local measurements ($\omega_i$) and the set-points ($P_j^{set}$) at the neighboring GFMs, denoted by the set $\mathcal{N}_i$\,.
The gains $\zeta_1\!>\!0\,,\, \zeta_2\!>\!0$ are tuned to regulate the learning rate of the gradient descent step and the speed of the frequency regulation.
The consensus-based secondary control updates the active power set-points on a slow timescale, at regular intervals $\Delta t$ (using a zero-order hold), as follows:
\begin{align}\label{eq:consensus}
    &P^{{set}}_{i,con}(t) = \begin{cases}
        P^{{set}}_{i,+} \;\mbox{as per \eqref{consPset} } & \forall t =k\,\Delta t\,,\,\text{$k$: integers} \\
        P^{{set}}_{i}(t-\Delta t) & \; \mbox{otherwise}
     \end{cases}
\end{align}

The control set-point adjustment, given in (\ref{consPset}), captures the frequency restoration to the nominal value and coordinated active power sharing criteria.  However, neither the droop-based primary nor the consensus-based secondary control guarantees a safe transient evolution of frequency within any pre-specified safety boundary $[\omega_{min},\omega_{max}]$ around the nominal frequency $\omega_0$\,. Additionally, the timescale difference between primary and secondary controls aggravates this issue in case of intermittent disturbances. We propose to solve these issues using a `safety-consensus' method, 
as described below. 

\subsection{Local Safety-Enforcing Control}
The primary objective of the safety-enforcing layer, as shown in the authors' recent work \cite{slac3r}, is to modify $P^{set}_i$  for each GFM-$i$, such that frequency satisfies safety limits $\omega_{min} \!\leq\! \omega_i\!\leq\! \omega_{max}$ throughout the pre- to post-disturbance transient evolution. The local, minimally invasive, safety control is achieved by defining barrier functions \cite{bouvier2022distributed,slac3r}:
\begin{align}\label{eq:def_h}
    h_{min}(\omega) := \omega -\omega_{min}\,,\,\text{ and }\,h_{max}(\omega) := \omega_{max}-\omega\,. 
\end{align}
The barrier-based safety control principle works as follows. The function $h_{min}$ (equivalently, $h_{max}$) takes negative values when the lower (equivalently, the upper) frequency limit is violated during transients. Therefore, whenever such \textit{unsafe} frequency excursions occur, the goal of the barrier-based safety-control is to ensure that the time-derivative of the corresponding barrier function ($h_{min}$ or $h_{max}$), along the system trajectories, is positive, i.e., $\dot{h}_{min}(\omega)\!\geq\! 0 $ or $\dot{h}_{max}(\omega) \!\geq\! 0 $\,.
%
%
To achieve faster and finite time steering of $\omega$ in the safe region, we enforce stricter conditions: 
\begin{align}\label{eq:h_dot}
    \!\!\!\dot{h}_{min}(\omega) &\!\geq\! -\alpha \left({h}_{min}(\omega)\right)^p,\,\,\dot{h}_{max}(\omega) \!\geq\! -\alpha \left({h}_{max}(\omega)\right)^p
\end{align} 
where $\alpha>0$ is a control gain and $p$ is an odd integer. Note that both $h_{min}$ and $\left(h_{min}\right)^p$ hold the the same signs, and so do $h_{max}$ and $\left(h_{max}\right)^p$. 
%
Using the definitions \eqref{eq:def_h} and the GFM dynamics \eqref{eq:gfm}, the time-derivative conditions in \eqref{eq:h_dot} yield the following safety-enforcing conditions on the set-points:
%
\begin{subequations}\label{des_con}
\begin{align}
  \omega_0 - \omega_i +m_{p_i}(P^{{set}}_i - P_i) \geq -\tau{\alpha} (\omega - \omega_{min})^p \\
  \omega_0 - \omega_i +m_{p_i}(P^{{set}}_i - P_i) \leq -\tau{\alpha} (\omega - \omega_{max})^p
\end{align}
\end{subequations}
Setting $\bar{\alpha} \!:=\! {\tau\alpha}/{m_p}$, we obtain the safety-control policy:
\begin{subequations}\label{up_low}
\begin{align}
    \!\!\!\!\textit{(safety) }\,\,P^{{set}}_{{i,low}}&\leq P^{set}_i\leq P^{{set}}_{{i,up}} \\
    \text{where, }P^{{set}}_{{i,low}} &:= P_i \!+\! \frac{1}{m_p} \!\left(\omega_i\!-\!\omega_0\right)\!\! -\bar{\alpha} (\omega_i\!-\!\omega_{{min}})^p \label{eq:P_set_low_eq}
    \\ 
    P^{{set}}_{{i,up}}  &:= P_i \!+\! \frac{1}{m_p} \!\left(\omega_i\!-\!\omega_0\right)\!\! -\bar{\alpha} (\omega_i\!-\!\omega_{{max}})^p \label{eq:P_set_up_eq}
\end{align}
\end{subequations}

\subsection{Proposed Enhancement: Safety-Consensus}

In our proposed hierarchical `safety-consensus' method, depicted in Fig.~\ref{fig:framework}, we augment the local, faster timescale, safety control law in \eqref{up_low}, with the distributed, slower timescale, consensus control law in \eqref{eq:consensus} to achieve the advantages of both methods, i.e., achieve three distinct objectives: (a) safe transient frequency evolution, (b) minimizing frequency deviation, and (c) coordinated power sharing. To do so, we integrate the obtained safety criteria on $P^{set}_{i,low}$ and $P^{set}_{i,up}$ with the consensus-based update $P^{{set}}_{i,con}(t)$\,, as follows:
\begin{align}\label{eq:P_set_saf}
        & \widehat{P}^{{set}}_{{i}}(t) = \min(P^{{set}}_{{i,up}}(t), \max(P^{{set}}_{{i,low}}(t), P^{set}_{i,con}(t)))    
\end{align}
Next, we summarize the main result and design guidelines concerning the safety-consensus control. Assume symmetric safety limits: $\omega_{min}\!:=\!\omega_0\!-\!\Delta\omega,\,\omega_{max}\!:=\!\omega_0\!+\!\Delta\omega$\,, for $\Delta \omega\!>\!0$\,. 
\begin{theorem}[Safety-Consensus]
    Consider sufficiently small droop gains $m_{p_i}\!<\!\Delta\omega/S_i$, and bounded consensus control set-points $\left|P_{i,con}^{set}\right|\!<\!S_i\!-\!\delta_i/m_{p_i}$, where
    \begin{align*}
        \delta_i:=\left(\frac{2\Delta\omega}{\alpha m_{p_i}}\right)^{1/p}.
    \end{align*}
    Then, for all bounded network disturbance $\left|P_i\right|\!\leq\!\Delta P$, for some $\Delta P\!\in\!(0,\,\Delta\omega/m_{p_i}\!-\!S_i)$, the proposed safety-consensus control satisfies the following properties: 
    \begin{enumerate}
        \item \textit{(minimally invasive):} whenever inside the safe region, i.e., $\!\left|\omega\!-\!\omega_0\right|\!\!\leq\!\Delta\omega$, the safety-consensus retains the consensus set-points, i.e., $\widehat{P}^{{set}}_{{i}}(t)=P^{set}_{i,con}(t)$\,.
        
    \item \textit{(safety-enforcing):} whenever the safety is violated by some margin $\delta_i$, i.e., $\left|\omega\!-\!\omega_0\right|\!>\!\Delta\omega+\delta_i$, safety-consensus overrides the consensus set-points, i.e.,
    \begin{align*}
        \omega(t)\!>\!\omega_0\!+\!\Delta\omega\implies \widehat{P}^{{set}}_{{i}}(t)\!=\!P^{{set}}_{{i,up}}(t)\implies \dot{\omega}\!<\!0\\
        \omega(t)\!<\!\omega_0\!-\!\Delta\omega\implies \widehat{P}^{{set}}_{{i}}(t)\!=\!P^{{set}}_{{i,low}}(t)\implies \dot{\omega}\!>\!0
    \end{align*}
    ensuring finite-time convergence to $\delta_i$-margin of safety.
    \end{enumerate}
\end{theorem}
\begin{proof}
    (Sketch) The first part of the claim follows by showing that $P_{i,con}^{set}\!\in\!(P_{i,low}^{set},P_{i,up}^{set})$ for every $\left|P_i\right|\!\leq\!\Delta P$\,, $\left|\omega\!-\!\omega_0\right|\!\leq\!\Delta\omega$\,, and $\left|P_{i,con}^{set}\right|\!\!<\!S_i\!-\!\delta_i/m_{p_i}$. For the second part, we need to show that $P_{i,con}^{set}\!\notin\!(P_{i,low}^{set},P_{i,up}^{set})$ whenever $\left|\omega\!-\!\omega_0\right|\!>\!\Delta\omega\!+\!\delta_i$\,, followed by combining \eqref{eq:h_dot} with \eqref{eq:gfm}.
\end{proof}
Finally, the capacity limit of the storage is enforced as:
\begin{align}
    \!\!P^{{set}}_{{i}}(t) \!&=\! \min\!\left(P^{{set}}_{{i,max}}(t), \max\!\left(-P^{{set}}_{{i,max}}(t), \widehat{P}^{{set}}_{i}(t)\right)\right)    \label{eq:P_set_fin}\\
    \text{with }\,&P^{set}_{i,max} := \sqrt{S_i^2 - Q_i^2}\,.\notag
\end{align}

\section{Simulation Results}
\begin{figure*}[t]
    \centering
    \includegraphics[width = 0.75\linewidth]{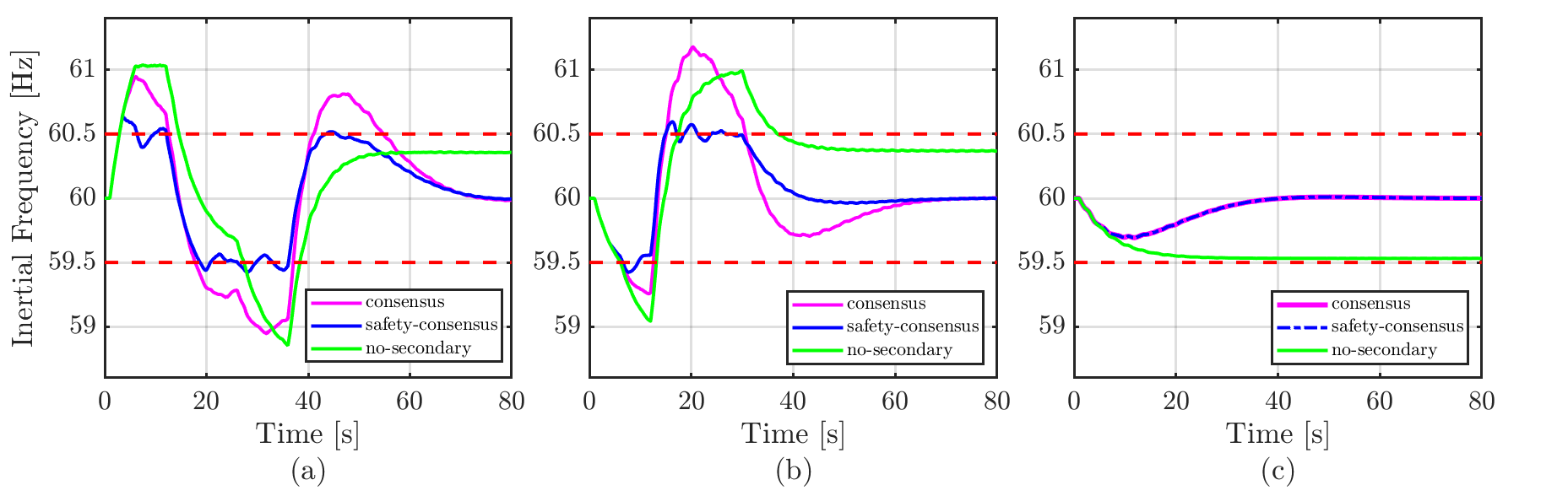} 
    \caption{{Grid inertial frequency response for three control methods under (a) Scenario 1, (b) Scenario 2, and (c) Scenario 3.} }
    \label{fig:comb_sc_freq}
    \vspace{-10pt}
\end{figure*}
We consider an embedded GFM-storage network on a modified IEEE 68-bus test system, with $m=16$ SGs, and $n=35$ GFM-storage units located at the load buses. Each storage capacity is considered to be uniform, with their cumulative capacity set at 10\% of the total system load \cite{chatterjee2024grid}. SGs and GFMs are modeled as per \eqref{eq:sg} and \eqref{eq:gfm}, respectively. 
Generator parameters are taken from the standard data set. For inverter parameters, the normalized value of the droop coefficients $m_i^p$ and $m_i^q$ are selected as 0.05 p.u./p.u. \cite{du_wei_mod}, while the filter time constant $\tau$ is set to $0.01s$. The secondary control layers parameters are as follows: $\zeta_1 = 2 $, $\zeta_2= 0.05$, $\bar{\alpha}=5\times10^6$, and  $p=3$. 
The GFM-storage units are expected to provide support during contingencies, and assumed to have no steady-state power injections. 
The lower and upper-frequency limits are defined as 59.5 and 60.5 Hz, respectively, while the nominal frequency value is 60 Hz. 
%
%
%
Three different controls are implemented for comparison: (a) `no-secondary' (only droop-based primary layer), (b) `consensus' (droop-based primary with consensus-based secondary), and (c) `safety-consensus' (proposed enhancement). The operation timescale of different control layers is furnished below:
\begin{table}[h!]
\centering
\begin{tabular}{c||cc}
\hline
\textbf{`no-secondary'}           & \multicolumn{1}{c||}{\textbf{`consensus'}} & \textbf{`safety-consensus'}         \\ \hline
\multirow{2}{*}{0.05 sec} & \multicolumn{1}{c||}{\multirow{2}{*}{4 sec}} & safety: 0.05 sec \\ \cline{3-3} 
                         & \multicolumn{1}{c||}{}                   & consensus: 4 sec \\ \hline
\end{tabular}
\end{table}
To illustrate its benefits, 
the performance of the proposed `safety-consensus' method is tested for two different large contingency scenarios (Scenarios 1, 2), each comprising multiple events. Moreover, a third, less severe, scenario (Scenario 3) is used to demonstrate the minimally-invasive property of the safety-consensus method. 

\textbf{Scenario 1.} At $t=1\,s$, load at bus 47 is decreased by 10.8\% (of total load). Next, the generator at bus 67 is tripped $t=6\,s$. During the ongoing system recovery, load at bus 18 and bus 4 are increased, respectively, by 13.5\% (of total load) at $t=12\,s$ and by 10.8\% (of total load) at $t=26\,s$. Finally, the system experiences a load decrease at bus 15 load by 15.2\% (of total load) at $t=36\,s$. 
Here, we computed inertial frequency (center of inertia (CoI) frequency) to describe the system frequency dynamics. \underline{\textsc{No-Secondary:}} According to Fig.~\ref{fig:comb_sc_freq}(a) (`no-secondary'), the load decrease at $1\,s$ caused the frequency to go above the safe operating limits of $60.5$ Hz, then the generator trip and two load increases caused the frequency to go below the lower limit of $59.5$ Hz. The final load change (a decrease) at $36\,s$ brings the frequency back to the safe operating region $[59.5, 60.5]$ Hz but with a steady-state deviation from the desired 60 Hz value. 
\underline{\textsc{Consensus:}} With the `consensus' method, the secondary control actions are effective in correcting the steady state frequency deviation with respect to the nominal frequency of 60 Hz; however, like the `no-secondary' case, the `consensus' method also fails to constrain the frequency value within the safe-zone of $[59.5, 60.5]$ Hz during the disturbance events. Particularly, after the load decreases at $1\,s$, the frequency transients fail to stay in the safe zone. 
Additionally, after the generator trip and two consecutive load increase frequency profile experiences a reverse excursion below 59.5 Hz with the consensus method and stays outside the safe operating zone for almost 20\,s. \underline{\textsc{Safety-Consensus:}} To this end, the proposed `safety-consensus' method successfully constrained the frequency within the safe operating zone throughout the contingency events. Particularly around $3\,s$ when the frequency crosses the upper bound $60.5$ Hz of the safe operating zone, safety control becomes active and helps the frequency to return to the safe zone. Likewise, at around $18\,s$\,, `safety-consensus' method limits the downward movement of frequency at $59.5$ Hz, and helps to maintain the lower safety limits throughout the contingency period till $35\,s$. Next, following the load decrease at $36\,s$, the proposed `safety-consensus' method successfully mitigates the upper frequency limit violations.
Finally, the in-built consensus-based component of the `safety-consensus' method effectively corrects the frequency deviation from 60 Hz after the transient phase passes. The set-point and active power injection profile are shown in Fig.~\ref{fig:set_point_1}. A heatmap capturing maximum frequency deviation from nominal frequency $\Delta f = |f-f_0|$ is plotted for each bus of the power network (see, Fig.~\ref{fig:freq_map_1}) for all 3 cases. Fig.~\ref{fig:freq_map_1} clearly shows that for `no-secondary' and `consensus', the frequency violation throughout the network is maximum, while the same is minimized for `safety-consensus'. 

\begin{figure}[t]
    \centering 
    \includegraphics[width = 1.00\linewidth]{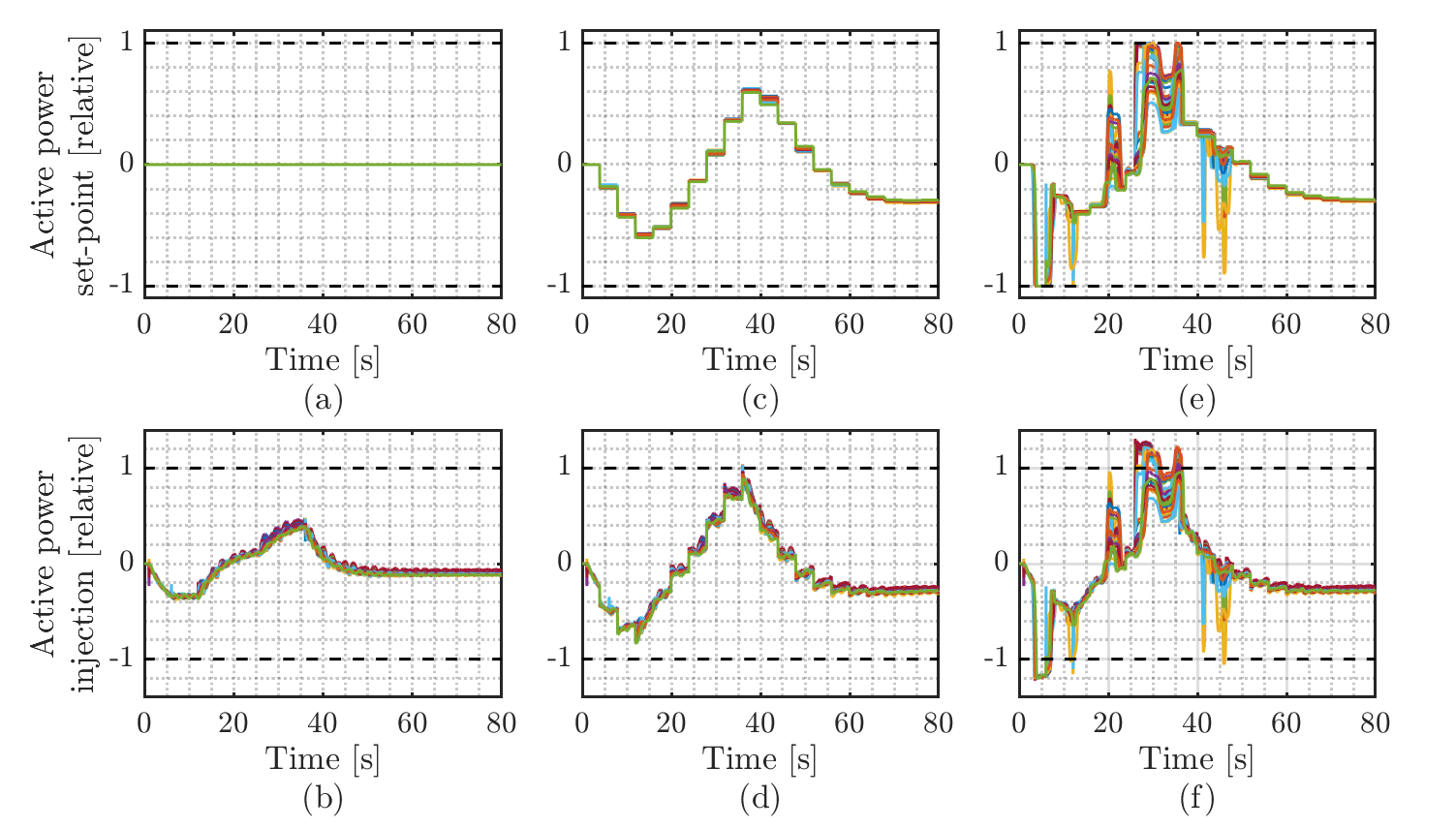} 
    \caption{Scenario 1: `no-secondary'--- active power (a) set-points, (b) injection; `consensus'--- active power (c) set-points, (d) injection; `safety-consensus'--- active power (e) set-points, (f) injection.} 
    \label{fig:set_point_1}
    \vspace{-10pt}
\end{figure}
\begin{figure}[t]
    \centering 
    \includegraphics[width = 1.00\linewidth]{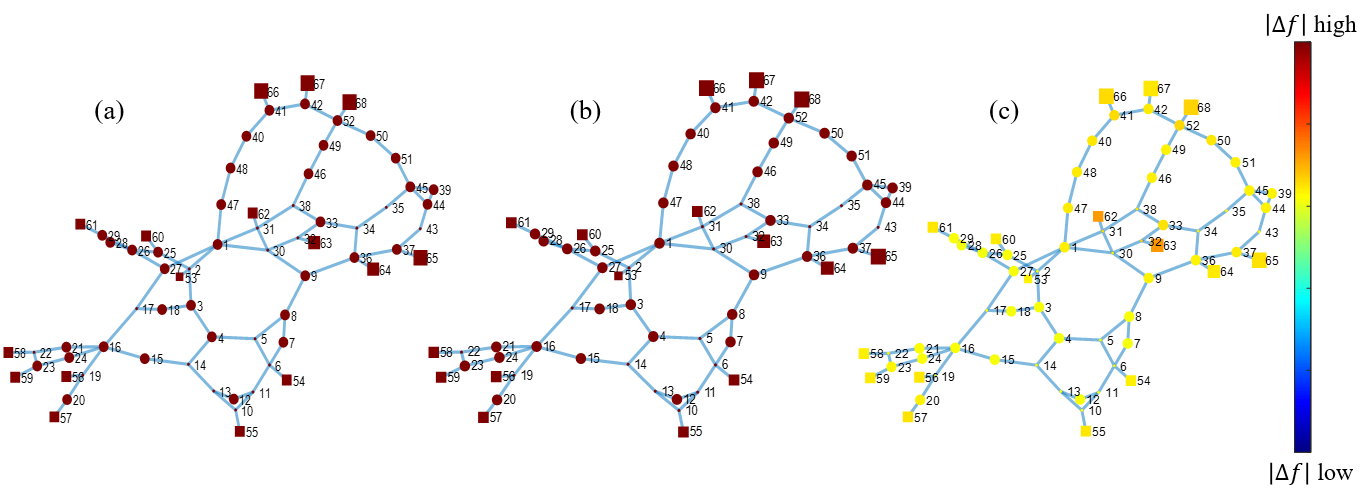} 
    \caption{Maximum Frequency deviation heatmap  for Scenario 1: (a) no-secondary, (b) consensus, (c) safety-consensus.}
    \label{fig:freq_map_1}
    \vspace{-10pt}
\end{figure}

\textbf{Scenario 2.} At $t=1\,s$, the generator at bus 67 is tripped, followed by the tripping of a second generator at bus 57 at $t=6\,s$. Next, during the ongoing system recovery, the load at bus 3 was reduced by 13\% (of total load) at $t=12\,s$, followed by a load increase at bus 44 by 9\% (of total load) at $t=30\,s$. Again, with `no-secondary' control, violations of both safety limits and steady-state deviation of frequency are observed in Fig.~\ref{fig:comb_sc_freq}(b). While with the `consensus' method, steady-state deviations are mitigated, safety limit violations still exist. On the contrary, our proposed `safety-consensus' method is successful in maintaining frequency within the safe operating zone of $[59.5, 60.5]$ Hz, and maintain \textit{zero $(0)$} steady-state frequency deviation, as shown in Fig.~\ref{fig:comb_sc_freq}(b). 

\textbf{Scenario 3.} To show that the `safety-consensus' method is minimally invasive, we tested a single disturbance scenario (\textit{Scenario 3}) with a load change of 3.8\% at bus 16. This disturbance does not cause any safety limit violation. Fig.~\ref{fig:comb_sc_freq}(c) clearly shows that the `safety-consensus' and `consensus' methods produce identical outcomes.

\section{Conclusion}
This paper considers the safety aspects in the control designs of GFM-interfaced storage units. The retirement of conventional generation necessitates the inclusion of renewable generations supported by two-way energy storage units for resilient operation. We integrated the safety enforcing control layer in the consensus-based implementation of designing secondary control of GFMs. A theoretical analysis is included to provide design guidelines for the proposed safety-consensus control, achieving the desired control objectives of minimally invasive, safe frequency regulation. Finally, we validate the proposed measurement-driven method with numerical experiments for the IEEE 68-bus test system.

\bibliographystyle{IEEEtran}
\bibliography{Ref_control_paper}

\end{document}